\documentclass[authoryear,preprint,review,12pt]{elsarticle}

\usepackage{amssymb}
\usepackage{amsthm}
\usepackage{mathrsfs}
\usepackage{amsmath}

\DeclareMathOperator{\EX}{\mathbb{E}}

\usepackage{lineno}

\newtheorem*{theorem}{Theorem}
\newtheorem{proposition}{Proposition}
\newtheorem{definition}{Definition}
\newtheorem*{remark}{Remark}
\newtheorem*{corollary}{Corollary}

\journal{Theoretical Population Biology}

\begin{document}

\begin{frontmatter}



\title{Drewnowski's index to measure lifespan variation: Revisiting the Gini coefficient of the life table}


\author[label1,label2]{Jos\'e Manuel Aburto}
\author[label3,label4]{Ugofilippo Basellini}
\author[label2]{Annette Baudisch}
\author[label2,label5,label6]{Francisco Villavicencio\corref{cor1}}

\cortext[cor1]{Corresponding author: fvillavicencio@ced.uab.es.}
\address[label1]{Leverhulme Centre for Demographic Science, Department of Sociology and Nuffield College, University of Oxford, Oxford OX1 1JD, UK.}
\address[label2]{Interdisciplinary Centre on Population Dynamics, University of Southern Denmark, 5230 Odense, Denmark.}
\address[label3]{Laboratory of Digital and Computational Demography, Max Planck Institute for Demographic Research, 18057 Rostock, Germany.}
\address[label4]{Mortality, Health and Epidemiology Unit, Institut national d'\'{e}tudes d\'emographiques (INED), 93322 Aubervilliers, France}
\address[label5]{Centre for Demographic Studies (CED), Universitat Aut\`onoma de Barcelona, 08193 Bellaterra, Spain.}
\address[label6]{Department of International Health, Johns Hopkins Bloomberg School of Public Health, Baltimore, MD 21205, USA.}


\begin{abstract}
The Gini coefficient of the life table is a concentration index that provides information on lifespan variation. Originally proposed by economists to measure income and wealth inequalities, it has been widely used in population studies to investigate variation in ages at death. We focus on a complementary indicator, Drewnowski's index, which is as a measure of equality. We study its mathematical properties and analyze how changes over time relate to changes in life expectancy. Further, we identify the threshold age below which mortality improvements are translated into decreasing lifespan variation and above which these improvements translate into increasing lifespan inequality. We illustrate our theoretical findings simulating scenarios of mortality improvement in the Gompertz model. Our experiments demonstrate how Drewnowski's index can serve as an indicator of the shape of mortality patterns. These properties, along with our analytical findings, support studying lifespan variation alongside life expectancy trends in multiple species.
\end{abstract}

\begin{keyword}
concentration index \sep Gompertz \sep life expectancy \sep lifespan inequality \sep mortality \sep threshold age



\end{keyword}

\end{frontmatter}



\section{Introduction}
\label{sec:intro}

The life table is an essential tool in mortality studies. It represents the current mortality experience of a population and it is usually summarized by life expectancy at birth $(e_o)$. Life expectancy at birth is the average years a synthetic cohort of newborns is expected to live if individuals were to experience the current mortality conditions throughout their lives. It is, however, an average indicator that masks variability in ages at death, which can be substantial. This variability is often referred to as lifespan variation or lifespan inequality, and has received increasing attention over the past two decades (see, to mention a few, \citealt{Wilmoth1999, edwards2005inequality, Smits2009, Baudisch2011, Vaupel2011, Fernandez2015, colchero2016emergence, ebeling2018rectangularization, vanraalte2018Science, permanyer2019global, aburtoDynamics, VaupelLongevityPNAS}). Lifespan variation reveals the uncertainty about the eventual age at death at the individual level, and measures how evenly mortality conditions are shared at the population level. There exist several indicators to measure lifespan variation (for an overview, see \citealt{Shkolnikov2003, vanRaalte2013}), such as the entropy of the life table \citep{leser_variations_1955, Keyfitz1977, demetrius1978adaptive}, the standard deviation or variance of the age-at-death distribution (as applied in \citealt{tuljapurkar2011variance}), the coefficient of variation (as applied in \citealt{Aburto2018Eastern}; Aburto, Wensink et al. \citeyear{aburto2018potential}), years of life lost $(e^\dagger)$ \citep{Goldman1986, Vaupel1986, Hakkert1987, Vaupel2003}, or the Gini coefficient \citep{Hanada1983}.

Here we study in greater detail the Gini coefficient of the life table $(G)$ and its complementary measure, Drewnowski's index $(D)$ \citep{Drewnowski1982, Hanada1983}, from a formal demographic perspective. We additionally aim to understand how changes in age-specific mortality underpin trends in lifespan variation. We focus on how changes over time in $D$ relate to changes in $e_o$, and highlight a new measure of absolute variation related to perturbation theory, named $\vartheta$. We provide the mathematical foundation of how Drewnowski's index (and analogously the Gini coefficient) evolves over time, and give analytical formulae to find the threshold age below which mortality improvements are translated into decreasing lifespan variation and above which these improvements translate into increasing lifespan inequality.


\section{The Gini coefficient and Drewnowski's index}
\label{sec:GiniCoef}

The Gini coefficient is one of the most popular indices employed in social sciences to measure concentration in the distribution of a non-negative random variable \citep{gini1912variabilita, gini1914sulla}. Originally proposed by economists to measure income or wealth inequality, this coefficient has been applied in demography and survival analysis to investigate within-group inequality in terms of ages at death (see, for instance, \citealt{Hanada1983, Shkolnikov2003, bonetti2009gini, gigliarano2017longevity, jones2018complexity, diaz2018mortality, basellini2019modelling, aburtoDynamics}).

\subsection{Definition}

As thoroughly discussed by \cite{yitzhaki2013gini}, there are several equivalent ways to define the Gini coefficient. Let $X$ be a non-negative random variable with probability density function $f(x)$ and expected value $\EX[X]$, a common definition is

$$
  G = \frac{1}{2\EX[X]}\int_0^\infty \int_0^\infty |x_1 - x_2|\,f(x_1)\,f(x_2)\,dx_1\,dx_2\;.  
$$
Accordingly, if $X$ is a random variable of the ages at death in a population, the Gini coefficient expresses the average of absolute differences in individual lifespans relative to the mean length of life $\EX[X]$.

\cite{michettid1957}, and later \cite{Hanada1983}, suggested a re-formulation of the Gini coefficient in terms of the life table functions, given by
\begin{equation}
  \label{Gini_definition}
  G = 1-\frac{\int_0^\infty\ell(x,t)^2\,dx}{\int_0^\infty\ell(x,t)\,dx} = 1- \frac{\vartheta}{e_o}\;,
\end{equation}
where $\ell(x,t)$ is the life table survival function at time $t$, $e_o=\int_0^\infty\ell(x,t)\,dx$ the life expectancy at birth at time $t$, and $\vartheta=\int_0^\infty \ell(x,t)^2\,dx$ is the resulting life expectancy at birth of doubling the hazard at all ages. \cite{jones2018complexity} interpret $\vartheta$ as a measure of \emph{shared life expectancy}, that is, the average time two newborns at time $t$ are expected to survive together. For the purposes of this article, the definition of the Gini coefficient in~\eqref{Gini_definition} will be used in the following.

\subsection{Main properties}

The Gini coefficient takes values between 0 and 1, and can be interpreted as a \emph{measure of inequality}. A value of 0 denotes equality in ages at death, i.e.~when every individual in the population has the exact same length of life. The index increases approaching 1 as lifespans become more spread and unequal in the population. This makes the interpretation clear and intuitive: higher values correspond to greater within-group inequality in ages at death. 

An additional attractive feature of the Gini coefficient is that it fulfills three important properties for inequality indices \citep{sen1973economic, anand1983inequality}: (i)~it does not change if the number of individuals at each age at death is changed by the same proportion (\emph{population-size independence}); (ii)~it does not change if each individual lifespan is changed by the same proportion (\emph{scale independence}); and (iii)~it decreases if years of life are transferred from a longer to a shorter lived individual (\emph{Pigou-Dalton condition}). Note that property~(i) enables straightforward comparison between populations, including comparisons between different species \citep{Wrycza2015}. Furthermore, the coefficient is not too sensitive to redistributions at early ages of life, and it well reflects changes at adult ages \citep{Shkolnikov2003}. 

Being bounded between 0 and 1, the Gini coefficient can be readily transformed from a \emph{measure of inequality} into a \emph{measure of equality} of lifespans. From~\eqref{Gini_definition}, a complementary measure of lifespan equality immediately derives, the Drewnowski's index, defined as
\begin{equation}
\label{Drew_definition}
D = 1-G = \frac{\vartheta}{e_o}=\frac{\int_0^\infty\ell(x,t)^2\,dx}{\int_0^\infty\ell(x,t)\,dx}\;,
\end{equation}
and sharing all the important properties of $G$. According to \cite{Hanada1983}, it was first proposed by Jan Drewnowski on a working group on health indicators at the World Health Organization in the early 1980s \citep{Drewnowski1982}.


\section{Changes over time in Drewnowski's index}
\label{sec:time-derv}

In order to analyze changes over time in Drewnowski's index---or, equivalently, the Gini coefficient---we aim to find an analytical expression for the time derivative of $D$. In the following, a dot over a function will denote the partial derivative with respect to time, but variable $t$ will be omitted for simplicity.

\subsection{Relative derivative of $D$}

\begin{proposition} 
	Let $D=\vartheta\,/\,e_o$ be Drewnowski's index, where $\vartheta=\int_0^\infty\ell(x)^2\,dx$, $e_o=\int_0^\infty\ell(x)\,dx$ is the life expectancy at birth, and $\ell(x)$ the probability of surviving from birth to age $x$. Then, relative changes over time in $D$ are given by
	\begin{equation}
	\label{changes_lambda}
	\frac{\dot{D}}{D}=\frac{\dot{\vartheta}}{\vartheta} - \frac{\dot{e}_o}{e_o}\;.
	\end{equation}
\end{proposition}

\begin{proof}%
	Note that $D=\vartheta\,/\,e_o$ implies that  $D\,e_o-\vartheta=0$. Differentiating with respect to time yields
	\begin{eqnarray*}
		\dot{D}\,e_o + D\,\dot{e}_o - \dot{\vartheta} = 0\;. 
	\end{eqnarray*}
Solving for $\dot{D}$ and dividing both sides by $D$, we get \eqref{changes_lambda}. 
\end{proof}

Equation~\eqref{changes_lambda} decomposes relative changes in $D$ into relative changes of the shared life expectancy between two individuals $\vartheta$, and relative changes in the life expectancy at birth $e_o$.

\subsection{Time derivatives of $e_o$ and $\vartheta$}

\cite{Vaupel2003} showed that changes over time in life expectancy at birth are a weighted average of the total rates of mortality improvements, expressed as
\begin{equation}
  \label{ex.derivative}
  \dot{e}_o=\int_0^\infty\rho(x)\,w(x)\,dx\;.
\end{equation} 
Function $\rho(x)=-\dot{\mu}(x)\,/\,\mu(x)$ stands for the age-specific rates of mortality improvement, where $\mu(x)$ is the force of mortality (hazard rate) at age $x$. The weights $w(x)=\mu(x)\,\ell(x)\,e(x)$ are a measure of the importance of death at age $x$, where $e(x)=\int_x^\infty\ell(a)\,da\,/\,\ell(x)$ is the remaining life expectancy at age $x$. Following a similar approach, we aim to express the time derivative of $\vartheta$ as a weighted average of mortality improvements, but with different weights.

\begin{definition}\label{Dx}
	Let $\ell(x)$ be the probability of surviving from birth to age $x$. A measure of lifespan equality at age $x$ is Drewnowski's index conditional upon survival up to age $x$, defined as
	\begin{equation}
		D(x)=\frac{1}{\ell(x)}\,\frac{\int_x^\infty\ell(a)^2\,da}{\int_x^\infty\ell(a)\,da}\;.
		\label{eq:Dx}
	\end{equation}	
\end{definition}

\begin{proposition}
	Let $\vartheta=\int_0^\infty\ell(x)^2\,dx$, where $\ell(x)$ is the probability of surviving from birth to age $x$. Then, its partial derivative with respect to time can be expressed as 
	\begin{equation}
	\dot{\vartheta}=\int_0^\infty \rho(x)\,w(x)\,2\,\ell(x)\,D(x)\,dx\;,
	\label{der.vartheta}
	\end{equation}
	where $\rho(x)$ are the age-specific rates of mortality improvement, $w(x)$ the same weights defined in~\eqref{ex.derivative}, and $D(x)$ as defined in~\eqref{eq:Dx}.

	\label{prop1}
\end{proposition}

\begin{proof}
	Applying the chain rule, the derivative of $\vartheta$ with respect to time is simply 
	
	$$
	\dot{\vartheta} = \int_0^\infty2\,\ell(x)\,\dot{\ell}(x)\,dx\;.
	$$
	Using that $\dot{\ell}(x)=-\ell(x)\int_0^x\dot{\mu}(a)\,da$, and reversing the order of integration, we get

	\begin{equation*}
	\begin{split}
	\dot{\vartheta}	
		& =-2\int_0^\infty \ell(x)^2 \,\int_0^x\dot{\mu}(a)\,da\,dx=-2\int_0^\infty\dot{\mu}(a)\int_a^\infty\ell(x)^2\,dx\,da	\\
		& = 2\int_0^\infty \rho(x)\,\mu(x)\,\ell(x)\,e(x)\,\frac{\int_x^\infty\ell(a)^2\,da}{\int_x^\infty\ell(a)\,da}\,dx 		\\
		& = \int_0^\infty \rho(x)\,w(x)\,2\,\ell(x)\,D(x)\,dx\;,
	\label{der.vartheta2}
	\end{split}
	\end{equation*}
	where $w(x)=\mu(x)\,\ell(x)\,e(x)$, which proves~\eqref{der.vartheta}.
\end{proof}

\subsection{Changes over time in $D$ in terms of mortality improvements}

Equations~\eqref{ex.derivative} and~\eqref{der.vartheta} enable us to express changes over time in $D$ in terms of mortality improvements. Using~\eqref{Drew_definition} and subsequently replacing~\eqref{ex.derivative} and~\eqref{der.vartheta} in~\eqref{changes_lambda} yields

\begin{align}
  \label{eq.changes_D}
  \dot{D} \nonumber
    & = D\,\left(\frac{\dot{\vartheta}}{\vartheta}-\frac{\dot{e}_o}{e_o}\right)=\frac{\dot{\vartheta}-\dot{e}_o D}{e_o}							\\ \nonumber
    & = \int_0^\infty\rho(x)\,w(x)\,\frac{2\,\ell(x)\,D(x)-D}{e_o}\,dx 					\\
    & = \int_0^\infty \rho(x)\,w(x)\,W(x)\,dx\;.
\end{align}
This result shows that changes over time in $D$ (and analogously in $G$) are a total average of mortality improvements weighted by $w(x)\,W(x)$, where $w(x)=\mu(x)\,\ell(x)\,e(x)$ are the same weights as in~\eqref{ex.derivative} and
\begin{equation}
	W(x)=\frac{2\,\ell(x)\,D(x)-D}{e_o}\;.
	\label{eq:weights}
\end{equation}


\section{The threshold age}
\label{sec:threshold}


\subsection{Positive and negative contributions to lifespan equality}

Because Drewnowski's index is a measure of equality, $\dot{D}>0$ indicates that lifespan variation decreases over time, whereas $\dot{D}<0$ implies that lifespan variation increases over time. Equation~\eqref{eq.changes_D} can then be used to analyze the existence of a threshold age that separates \emph{positive} from \emph{negative} contributions to lifespan equality as a result of mortality improvements.

Note that in the assumption that mortality improvements occur at all ages, $\rho(x)=-\dot{\mu}(x)\,/\,\mu(x)>0$ is a strictly positive function. Therefore, from~\eqref{eq.changes_D},
\begin{enumerate}
	\item Those ages $x$ for which $w(x)\,W(x)>0$ will contribute \emph{positively} to Drewnowski's index $D$ and increase lifespan equality;
	\item Those ages $x$ for which $w(x)\,W(x)<0$ will contribute \emph{negatively} to Drewnowski's index $D$ and favor lifespan inequality;
	\item Those ages $x$ for which $w(x)\,W(x)=0$ will have no effect on the variation over time of $D$.
\end{enumerate}

Any existing threshold age that separates positive from negative contributions to lifespan equality will occur whenever $w(x)\,W(x)=0$. Since $\mu(x)$, $\ell(x)$, and $e(x)$ are all positive functions, so are $w(x)$ and $e_o$. Hence,
\begin{equation}
	\label{eq:gfunc}
	w(x)\,W(x)=0\quad\Longleftrightarrow\quad 2\,\ell(x)\,D(x)-D=0\;.
\end{equation}

\subsection{Existence and uniqueness of the threshold age}

By means of the following two propositions and one theorem, we aim to prove that in a scenario in which mortality improvements occur at all ages and $\rho(x)>0$ for all $x\geq0$, there exists a unique threshold age $a^D$ that separates positive from negative contributions to lifespan equality (measured by $D$) as a result of those improvements.

\begin{remark}
	Following~\eqref{Drew_definition}, Drewnowski's index $D$ is bounded between 0 and 1, reaching a value of 1 at complete equality in the ages at death within a population. A score of 0 would express maximum inequality in the ages at death. By definition, however, this value can never be reached:
	\begin{equation}
	\label{eq:d0}
	D=0\,\Longleftrightarrow\,\frac{\int_0^\infty\ell(x)^2\,dx}{\int_0^\infty\ell(x)\,dx}=0\,\Longleftrightarrow\,\int_0^\infty\ell(x)^2\,dx=0\,\Longleftrightarrow\,\ell(x)=0
	\end{equation}
	for all ages $x\geq0$. This implies that the denominator in~\eqref{eq:d0} equals 0 because $\ell(x)\geq0$ is always positive and, consequently,  $D$ would be undefined. Hence, $0<D\leq1$.
\end{remark}

\begin{proposition}
	\label{prop3}
	Let $\ell(x)$ be the probability of surviving from birth to age $x$, $D$ Drewnowski's index as defined in~\eqref{Drew_definition}, and $D(x)$ as defined in~\eqref{eq:Dx}. Define the function $g(x):=2\,\ell(x)\,D(x)-D$. Then, there exists at least one age $a^D$ such that $g(a^D)=0$.
\end{proposition}

\begin{proof}
	At age $x=0$,

	\begin{equation}
		\label{eq:limitlow}
		g(0)=2\,\ell(0)\,D(0)-D=2\,D-D=D>0
	\end{equation}
	by definition, since $0<D\leq1$.
	
	When ages become arbitrarily large,
	$$
	\lim_{x\to\infty}g(x)=\lim_{x\to\infty}(2\,\ell(x)\,D(x)-D)=2\lim_{x\to\infty}\ell(x)\,D(x)-D\;,
	$$
	which only depends on the behavior of $\ell(x)\,D(x)$. Because $\ell(x)\in[0,1]$ for all ages $x\geq0$, we have that $0\leq\ell(x)^2\leq\ell(x)$ and 

	$$
	0\leq\lim_{x\to\infty}\int_x^\infty\ell(a)^2\,da\leq\lim_{x\to\infty}\int_x^\infty\ell(a)\,da=\lim_{x\to\infty}e(x)\,\ell(x)=0\;,
	$$
	where $e(x)$ is the remaining life expectancy at age $x$. This proves that both integrals tend to 0 as $x$ approaches~$\infty$. Consequently,

	$$
	\lim_{x\to\infty}\ell(x)\,D(x)=\lim_{x\to\infty}\frac{\int_x^\infty\ell(a)^2\,da}{\int_x^\infty\ell(a)\,da}
	$$
	is indeterminate, but applying L'H\^opital's rule, we get

	$$
	\lim_{x\to\infty}\ell(x)\,D(x)=\lim_{x\to\infty}\frac{\frac{\partial}{\partial x}\,\int_x^\infty\ell(a)^2\,da}{\frac{\partial}{\partial x}\,\int_x^\infty\ell(a)\,da}=\lim_{x\to\infty}\frac{-\ell(x)^2}{-\ell(x)}=\lim_{x\to\infty}\ell(x)=0\;.
	$$

	As a result, 

	\begin{equation}
		\label{eq:limitup}
		\lim_{x\to\infty}g(x)=2\lim_{x\to\infty}\ell(x)\,D(x)-D=0-D<0\;.
	\end{equation}
	Finally, using~\eqref{eq:limitlow} and~\eqref{eq:limitup}, on a continuous framework the intermediate value theorem guarantees the existence of at least one positive age $a^D$ at which $g(a^D)=0$.
\end{proof}

\begin{proposition}
	\label{prop4}
	Let $\ell(x)$ be the probability of surviving from birth to age $x$, $D$ Drewnowski's index as defined in~\eqref{Drew_definition}, and $D(x)$ as defined in~\eqref{eq:Dx}. Then, $g(x):=2\,\ell(x)\,D(x)-D$ is a strictly decreasing function. 
\end{proposition}

\begin{proof}
	In order to demonstrate that $g(x)$ is a strictly decreasing function it suffices to show that its first derivative is negative for all ages $x\geq0$. Note that since $D$ does not depend on age,
	
	$$
	\frac{\partial}{\partial x}\,g(x)<0\quad\Longleftrightarrow\quad\frac{\partial}{\partial x}\,\big(\ell(x)\,D(x)\big)<0\;.
	$$

	Applying the quotient rule together with the fundamental theorem of calculus, we get

	\begin{equation*}
		\begin{split}
		\frac{\partial}{\partial x}\,\big(\ell(x)\,D(x)\big) & = \frac{\partial}{\partial x}\,\left(\frac{\int_x^\infty\ell(a)^2\,da}{\int_x^\infty\ell(a)\,da}\right) \\ & = \frac{\int_x^\infty\ell(a)\,da\,\frac{\partial}{\partial x}\,\left(\int_x^\infty\ell(a)^2\,da\right) - \int_x^\infty\ell(a)^2\,da\,\frac{\partial}{\partial x}\,\left(\int_x^\infty\ell(a)\,da\right)}{\left(\int_x^\infty\ell(a)\,da\right)^2} \\ & =
		\frac{\int_x^\infty\ell(a)\,da\,\left(-\ell(x)^2\right) - \int_x^\infty\ell(a)^2\,da\,\left(-\ell(x)\right)}{\left(\int_x^\infty\ell(a)\,da\right)^2}\,.
		\end{split}
	\end{equation*}
	Hence,

	\begin{equation*}
		\begin{split}
		\frac{\partial}{\partial x}\,g(x)<0
			& \,\Longleftrightarrow\,\ell(x)\int_x^\infty\ell(a)^2\,da-\ell(x)^2\int_x^\infty\ell(a)\,da<0		\\
			& \,\Longleftrightarrow\,\frac{1}{\ell(x)^2}\int_x^\infty\ell(a)^2\,da < \frac{1}{\ell(x)}\int_x^\infty\ell(a)\,da\;.
		\end{split}
	\end{equation*}
	
	Note that $\ell(x)=\exp\left[-\int_0^x\mu(a)\,da\right]$ for a given age-specific hazard function $\mu(x)$. Therefore, $\ell(x)^2=\exp\left[-\int_0^x2\,\mu(a)\,da\right]$	can be interpreted as the survival schedule with doubling hazard $2\,\mu(x)$ at all ages $x\geq0$. We can then define $\tilde{e}(x)=\int_x^\infty\ell(a)^2\,da\,/\,\ell(x)^2$ as the remaining life expectancy at age $x$ of a population with survival schedule $\ell(x)^2$ and age-specific force of mortality $2\,\mu(x)$. Accordingly, 

	$$
	\frac{\partial}{\partial x}\,g(x)<0\,\Longleftrightarrow\, 
	\frac{1}{\ell(x)^2}\int_x^\infty\ell(a)^2\,da < \frac{1}{\ell(x)}\int_x^\infty\ell(a)\,da\,\Longleftrightarrow\,\tilde{e}(x)<e(x)
	$$
	for all $x\geq0$, which holds true since doubling the hazard corresponds to a lower remaining life expectancy, in the reasonable assumption that $\mu(x)>0$ for all ages.
\end{proof}

\begin{theorem}
	Let $D=\vartheta\,/\,e_o$ be Drewnowski's index, where $\vartheta=\int_0^\infty\ell(x)^2\,dx$, $e_o=\int_0^\infty\ell(x)\,dx$ is the life expectancy at birth, and $\ell(x)$ the probability of surviving from birth to age $x$. Assume mortality improvements over time occur at all ages. Then, there exists a unique threshold age $a^D$ that separates positive from negative contributions to lifespan equality, measured by $D$, as a result of those improvements.
\end{theorem}

\begin{proof}
	Following~\eqref{eq.changes_D}, changes over time in $D$ can be expressed as a weighted average of mortality improvements, given by

	$$
	\dot{D}=\int_0^\infty\rho(x)\,w(x)\,W(x)\,dx\;,
	$$
	where $\rho(x)$ are the age-specific rates of mortality improvement over time, and $w(x)\,W(x)$ the weights. By assumption, $\rho(x)>0$ for all ages $x\geq0$. Therefore, any threshold age that separates positive from negative contributions to lifespan equality as a result of mortality improvements will occur whenever $w(x)\,W(x)=0$. From~\eqref{eq:gfunc},

	$$
	w(x)\,W(x)=0\quad\Longleftrightarrow\quad 2\,\ell(x)\,D(x)-D=0\;,
	$$
	where $D(x)$ is as defined in~\eqref{eq:Dx}. Proposition~\ref{prop3} guarantees the existence of at least one positive age $a^D$ such that $2\,\ell\left(a^D\right)\,D\left(a^D\right)-D=0$. In addition, from Proposition~\ref{prop4} the function $g(x):=2\,\ell(x)\,D(x)-D$ is strictly decreasing. Hence, assuming continuity, $g(x)$ is a one-to-one function and therefore the threshold age $a^D$ is unique. 
\end{proof}

\begin{corollary}
    
    Let $G$ be the Gini coefficient as defined in~\eqref{Gini_definition}. Provided that $G=1-D$, following~\eqref{eq.changes_D}
    
    $$
        \dot{G}=-\dot{D}=-\int_0^\infty \rho(x)\,w(x)\,W(x)\,dx\;.
    $$
    Hence, $G$ and $D$ have the same threshold age, which is unique in the assumption that mortality improvements occur at all ages. Improvements below the threshold age will reduce lifespan variation (lowering $G$, but enlarging $D$), and improvements above will increase lifespan inequality (enlarging $G$, but lowering $D$).    
 
\end{corollary}


\section{Application}
\label{sec:application}

We illustrate our theoretical findings simulating scenarios of mortality improvement depending on level versus rate changes in the Gompertz mortality model  \citep{gompertz1825}. This model captures an exponential change of the force of mortality $\mu(x)=\alpha\cdot\mathrm{e}^{\beta\,x}$ over age $x$, where $\alpha$ is the baseline level of mortality and $\beta$ the rate of aging. The R code \citep{Rsoftware} to reproduce all the figures and results discussed in the following will be publicly available upon acceptance.

\subsection{Gompertz force of mortality with positive aging: level vs rate change}

Our first scenario demonstrates how improvements in mortality affect our outcome variables when progress results from reducing the level of mortality, but not from slowing the rate of aging. Figure~\ref{Fig_1} depicts the Gompertz mortality function on a log-scale (Panel~A), the corresponding age-at-death distribution (Panel~B), and associated weights $W(x)$ (Eq.~\ref{eq:weights}) (Panel~C), for a fixed rate of aging $\beta =0.1$ and changing levels of baseline mortality $\alpha$. Values of Drewnowski's index $D$ and corresponding threshold ages $a^D$ are reported in Panels~B) and~C).

\begin{figure}[ht]
	\begin{center}
		\includegraphics[width=\textwidth]{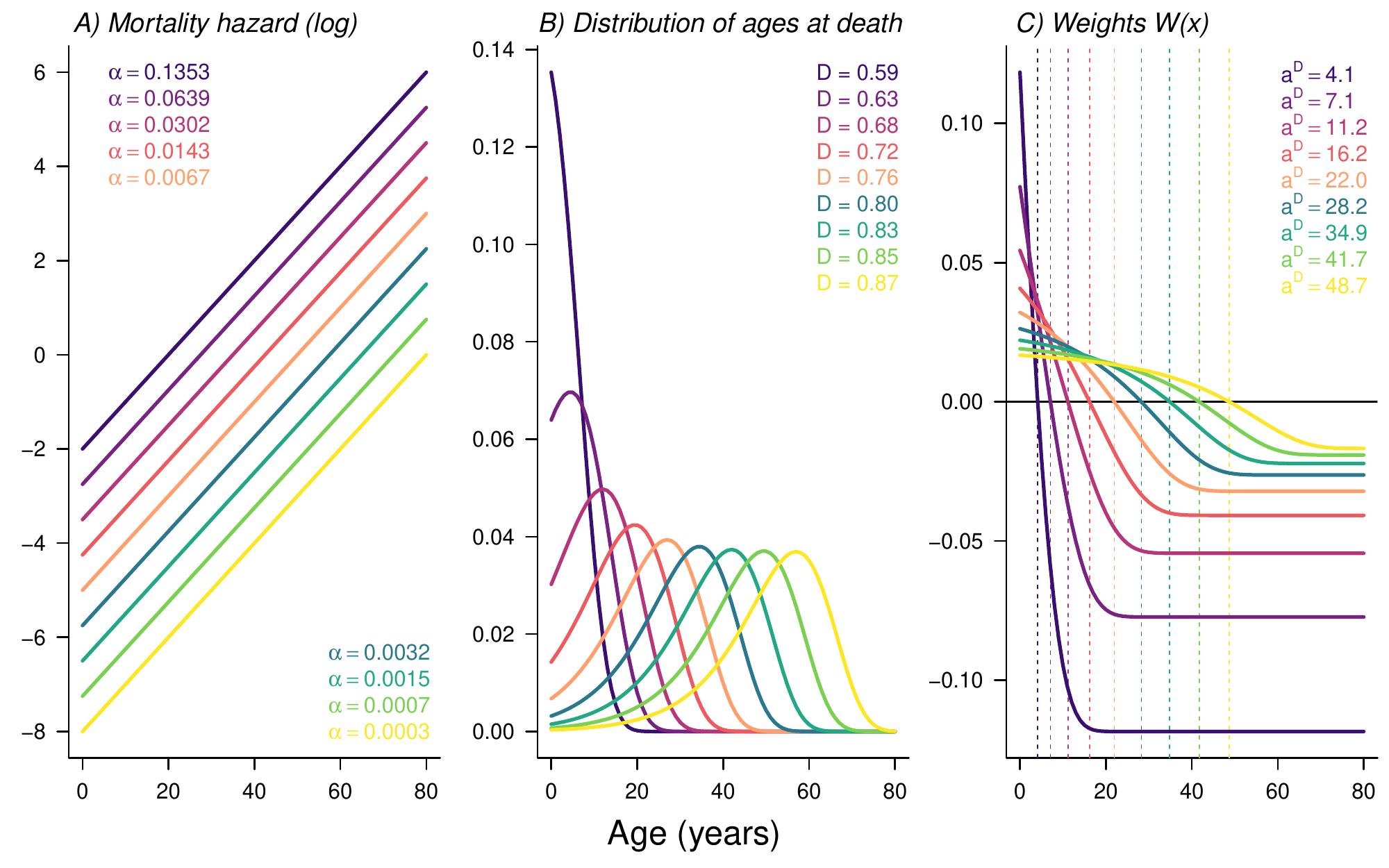}
		\caption{Gompertz mortality model with positive aging $(\beta>0)$ for different levels of baseline mortality $\alpha$ and a fixed rate of aging $\beta =0.1$}
		\label{Fig_1}
	\end{center}
\end{figure}

Lower levels of baseline mortality $\alpha$ translate into vertical downward shifts in the log-hazard (Panel~A). Correspondingly, lifespan variation falls as deaths concentrate at older ages, as captured by larger values of  $D$ (Panel~B). Lower levels of mortality raise the threshold age $a^D$ at which the weights $W(x)$ switch sign from positive to negative (Panel~C). Positive weights at some age imply that saving lives at that age increases equality, i.e. enlarging $D$; negative weights imply that saving lives decreases equality, i.e. lessening $D$. The threshold age  $a^D$ marks the boundary between these positive and negative effects of life saving on lifespan equality and is indicated by the dashed vertical lines. Panel~C) shows that under the assumption of mortality improvements over all ages, this threshold age is unique and that the magnitude of the weights $W(x)$ over age changes depending on the different shapes of the distribution of deaths, supporting our theoretical results. 

Complementing  these findings, Fig.~\ref{Fig_2} analyzes how slowing the rate of aging $\beta$ affects lifespan variation. For a fixed baseline level of mortality $\alpha=\mathrm{e}^{-5}\approx0.0067$, Panel~A) illustrates how lower values of $\beta$ translate, as expected,  into a slower rise in the log-hazard. In contrast to the first scenario, Panel~B) reveals that slowing the rate of aging  $\beta$ reduces $D$ (lighter colors) and widens the distribution of ages at death. That is, extending lifespan by slowing the rate of aging increases lifespan inequality, which is consistent with previous research by \cite{ColcherEtAl2021}. The threshold age $a^D$ increases as lifespan equality decreases. While life saving has opposite effects on lifespan variability and the level of the threshold age in the two scenarios, the general dynamics of the weights $W(x)$ remain similar. Saving lives before the threshold age reduces inequality; saving lives later increases inequality. The threshold age $a^D$ remains unique, independent of specific values of the mortality parameters, further supporting our theoretical predictions. 

\begin{figure}[ht]
	\begin{center}
		\includegraphics[width=\textwidth]{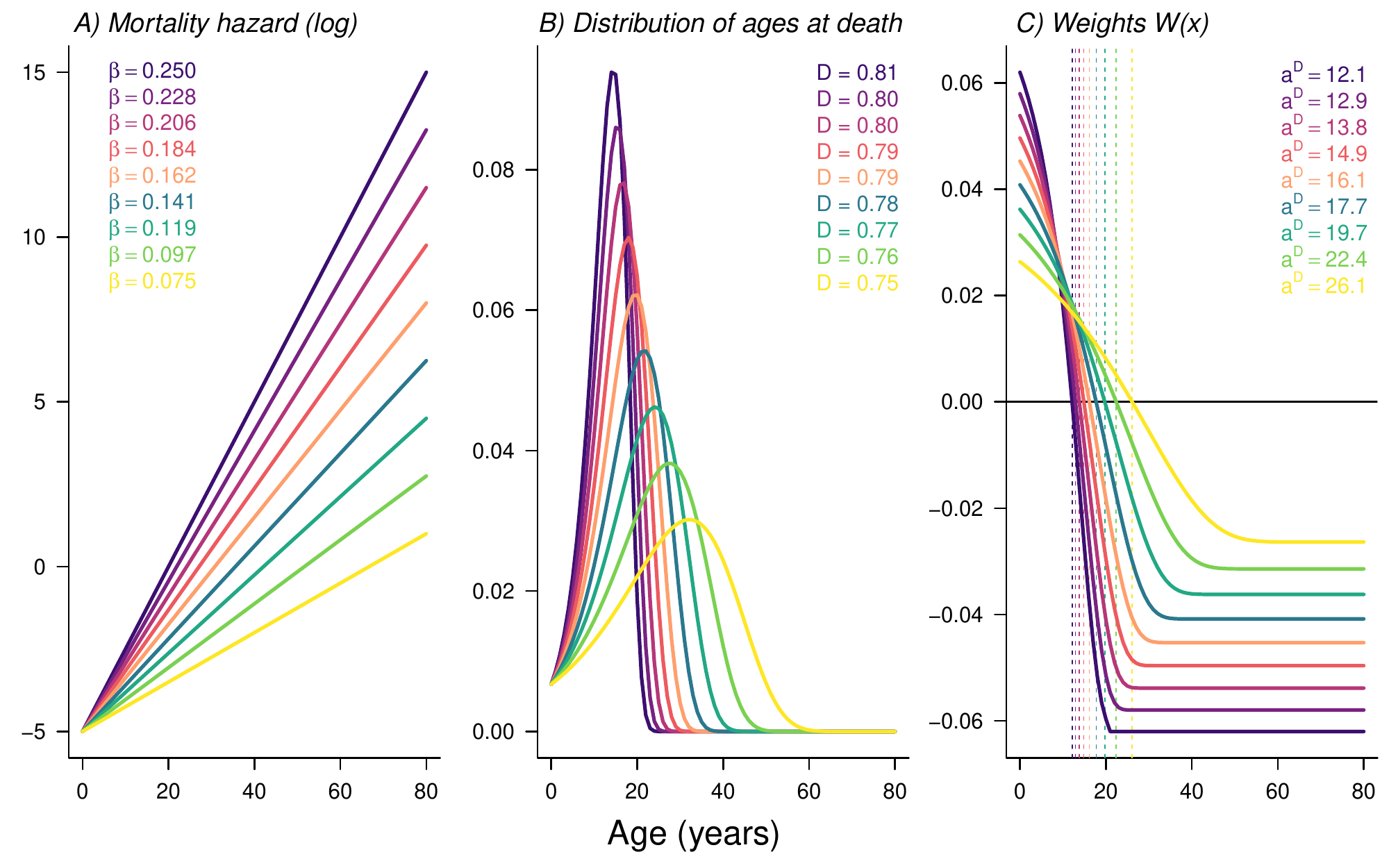}
		\caption{Gompertz mortality model with positive aging for different rates of aging $\beta>0$ and a fixed level of mortality $\alpha=\mathrm{e}^{-5}$}
		\label{Fig_2}
	\end{center}
\end{figure}

\subsection{Zero rate of aging}

Figure~\ref{Fig_3} captures a scenario in which the rate of aging equals zero. Under varying but constant levels of mortality (Panel~A), lifespan variation remains unchanged at $D=0.5$ and the largest number of deaths occur at the earliest ages (Panel~B). The lower the level of mortality, the longer the length of life, and the later the threshold age (Panel~C). 

This can also be verified analytically. When $\beta=0$, the survival function for the Gompertz mortality model simplifies to $\ell(x)=\mathrm{e}^{-\alpha\,x}$. Then, life expectancy as the integral of the survivorship over all ages becomes $1/\alpha$, which is the inverse of the force of mortality. Following Definition~\ref{Dx}, in this scenario the conditional Drewnowski index $D(x)$ is

$$
D(x)=\frac{1}{\mathrm{e}^{-\alpha\,x}}\,\frac{\mathrm{e}^{-2\,\alpha\,x}\,/\,2\,\alpha}{\mathrm{e}^{-\alpha\,x}\,/\,\alpha}=0.5\;.
$$
In other words, when mortality is constant over age, $D(x)=0.5$ for all ages and it is independent of the baseline mortality level $\alpha$. Accordingly, from Proposition~\ref{prop3}

$$
	g(x)=0\Longleftrightarrow 2\,\mathrm{e}^{-\alpha\,x}\,0.5-0.5=0\Longleftrightarrow x=\ln(2)\,/\,\alpha\;,
$$
which implies that the threshold age is $a^D=\ln(2)\,\,/\alpha$ and  $\ell(a^D)=\mathrm{e}^{-\alpha\,a^D}=0.5$. This means that, when mortality is constant, the threshold age and the median age at death are the same.
\smallskip

\begin{figure}[ht]
	\begin{center}
		\includegraphics[width=\textwidth]{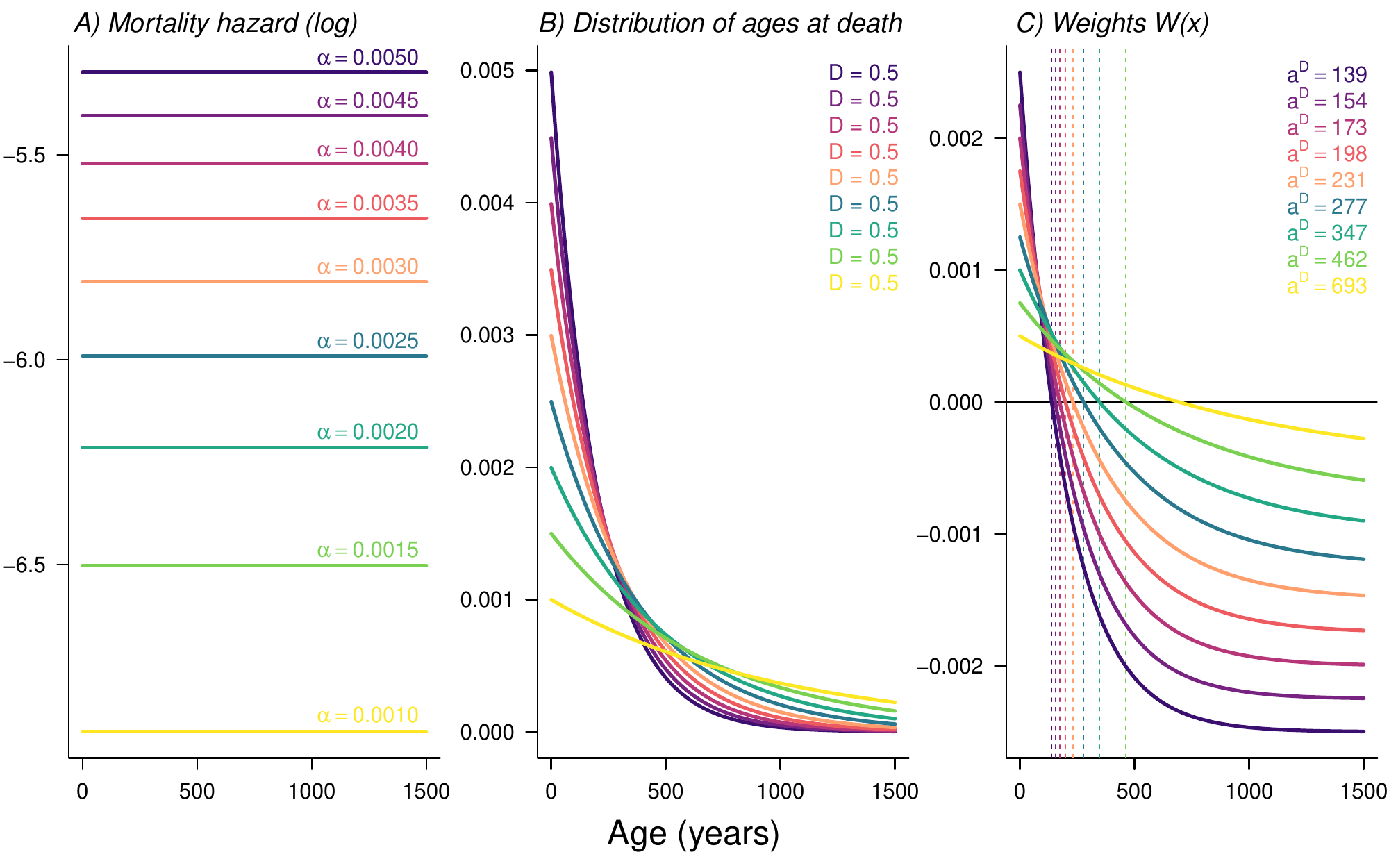}
		\caption{Gompertz mortality model with zero rate of aging $(\beta = 0)$ and different levels of baseline mortality $\alpha$}
		\label{Fig_3}
	\end{center}
\end{figure}

\subsection{Negative rate of aging}

To complete the analyses, Figs.~\ref{Fig_4} and~\ref{Fig_5} capture scenarios of decreasing mortality, which are consistently marked by values of $D$ below one half (Panels~B). Different to the scenario with positive aging rates, here effects of both level and rate changes point in the same direction: higher mortality results in less lifespan variation (higher values of $D$ in Panels~B). Even though the value of the Drewnowski's index varies substantially (from $0.03$ to $0.48$), the threshold age barely changes and remains close to the earliest age (Panels~C). This is because in these scenarios most individuals die young, and life expectancy is extremely short for all cases.

The two negative aging scenarios substantially differ, however, in the maximum age for the survivors (not shown). This reflects the exceeding disparity among the majority of individuals who die upon---or shortly after---birth, and those few surviving to manifold higher ages as they benefit from the improvements in mortality over age. These differences are also reflected in the density functions (Panels~B). While varying the rate of aging leaves no distinguishable effect on the density function (Fig.~\ref{Fig_5}, Panel~B), varying the baseline level of mortality $\alpha$ distinctly affects the shape of the density function with lower mortality levels yielding an increasing rectangularization. 

\begin{figure}[ht]
	\begin{center}
		\includegraphics[width=\textwidth]{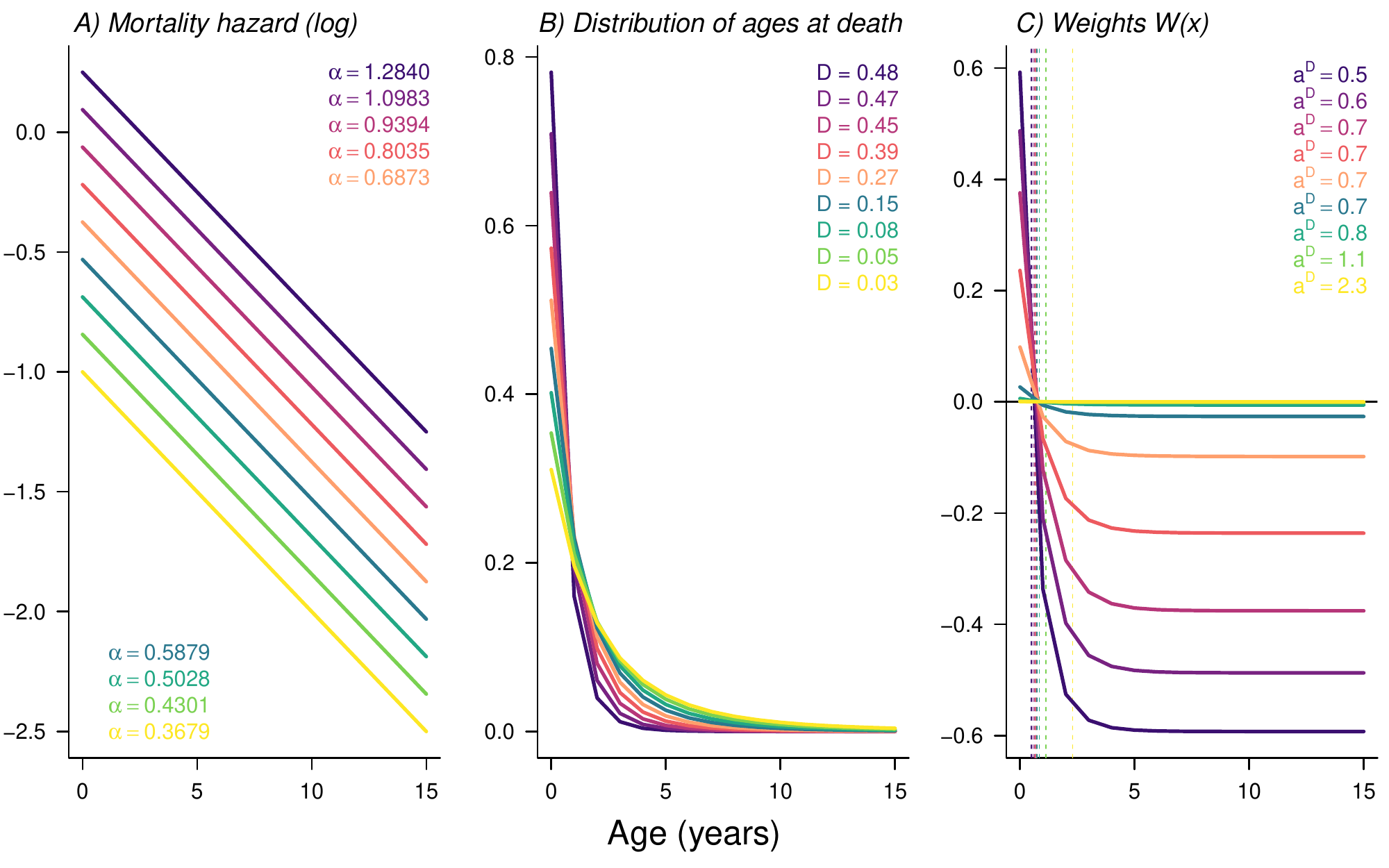}
		\caption{Gompertz mortality model with negative rate of aging $(\beta<0)$ for different levels of baseline mortality $\alpha$ and a fixed rate of aging $\beta = -0.1$}
		\label{Fig_4}
	\end{center}
\end{figure}

\begin{figure}[ht]
	\begin{center}
		\includegraphics[width=\textwidth]{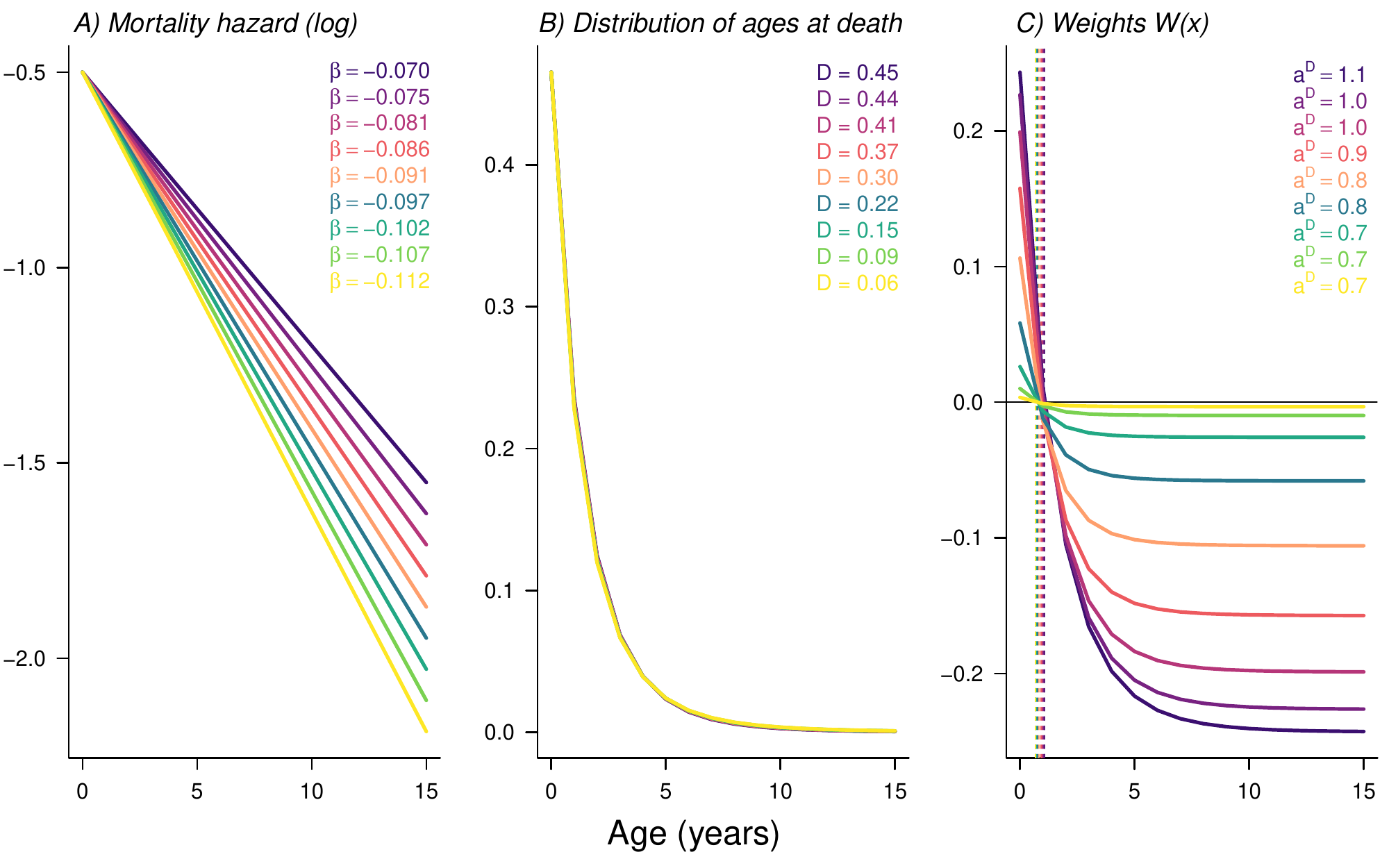}
		\caption{Gompertz mortality model with negative rate of aging ($\beta<0$) for different rates of aging $\beta$ and a fixed level of mortality $\alpha=\mathrm{e}^{-0.5}$}
		\label{Fig_5}
	\end{center}
\end{figure}


\section{Conclusion}

In this article we uncovered how age-specific patterns of mortality underpin trends in lifespan variation as measured by Drewnowski's index~$D$, a complementary indicator to the Gini coefficient of the life table, by means of formal demography. We contribute to the literature by disentangling how changes in age patterns of mortality affect lifespan variation and provide an analytical proof for the existence of a threshold age $a^D$ below which mortality improvements are translated into increasing lifespan equality and above which these improvements translate into increasing lifespan inequality. Previous research determined such age for the life table entropy \citep{aburto2019threshold}, the variance \citep{Gillespie2014} and $e^\dagger$ \citep{Zhang2009}.

We test and illustrate our results under multiple scenarios of mortality changes with age, including positive, zero and negative rates of aging. This is important because shapes of mortality patterns across the tree of life vary substantially (Jones, Scheuerlein et al. \citeyear{jones2014diversity}). For example, while some species such as humans exhibit increasing mortality over adult ages, other species experience an unchanged force of mortality over age, such as \emph{Hydra} (Schaible, Scheuerlein et al. \citeyear{schaible2015constant}), or even declining adult mortality, called `negative senescence'  \citep{vaupel2004case, baudisch2008inevitable, cayuela2019life}. Our experiments demonstrate how Drewnowski's index~$D$ can serve as an indicator of the `shape' of mortality patterns \citep{Baudisch2011, Wrycza2015}, distinguishing between increasing mortality $(D>0.5)$, constant mortality $(D=0.5)$, and decreasing mortality $(D<0.5)$ with age. These properties, along with our analytical findings, support studying lifespan variation alongside life expectancy trends in multiple species. 

\section*{Acknowledgements}

This work was partially funded by the European Research Council (Grant ERC-2019-AdG-884328) (J.M.A. and F.V.). J.M.A. also acknowledges funding from the British Academy's Newton International Fellowship (Grant NIFBA19/190679) and the Leverhulme Trust Large Centre, and F.V. from the Juan de la Cierva Programme, Spanish Ministry of Science and Innovation (Grant IJC2019-039144-I).



\bibliographystyle{elsarticle-harv} 
\bibliography{references}





\end{document}